\documentclass[12pt,a4paper]{article}
\usepackage{amsthm,amsfonts,amsmath,amssymb,bookmark}
\usepackage{color}

\theoremstyle{plain}
\newtheorem{theorem}{Theorem}
\newtheorem{lemma}[theorem]{Lemma}

\theoremstyle{definition}

\theoremstyle{remark}

\def\bq{\begin{eqnarray}}
\def\eq{\end{eqnarray}}
\def\bqq{\begin{eqnarray*}}
\def\eqq{\end{eqnarray*}}
\def\nn{\nonumber}

\def\eps{\varepsilon}

\def\R{\mathbb{R}}

\def\cE {\mathcal{E}}
\def \d {{\rm d}}

\title{Blow-up profile of the focusing Gross-Pitaevskii minimizer under self-gravitating effect}

\author{Thanh Viet Phan \\
\normalsize{Applied Analysis Research Group, Faculty of Mathematics and Statistics,} \\
\normalsize{Ton Duc Thang University, Ho Chi Minh City, Vietnam}\\
\normalsize{phanthanhviet@tdt.edu.vn} 
}

\date{\normalsize\today}

\begin{document}

\maketitle


\begin{abstract} We consider a Bose-Einstein condensate in a 2D dilute Bose gas, with an external potential and an interaction potential  containing both of the short-range attractive self-interaction and the long-range self-gravitating effect. We prove the existence of minimizers and analyze their behavior when the strength of the attractive interaction converges to a critical value. The universal blow-up profile is the unique optimizer of a Gagliardo-Nirenberg interpolation inequality. 

\bigskip

   
\end{abstract}


\section{Introduction}

Since the first observation in 1995 in the Nobel Prize winning works of Cornell, Wieman, and Ketterle \cite{CorWie-95,Ketterle-95}, the Bose-Einstein condensation has been studied intensively in the last decades. When the interaction is attractive, it is a remarkable that the condensate may collapse, as noted in experiments \cite{BraSacTolHul-95,SacStoHul-98,KagMurShl-98}. In the present paper, we will study this collapse phenomenon in a rigorous model. 

We consider a Bose-Einstein condensate in a 2D dilute Bose gas, with an external potential $V:\R^2\to \R$ and an interaction potential  $\omega(x-y)$ containing both of the short-range attractive self-interaction and the long-range self-gravitating effect. For simplicity, we take 
\bq \label{eq:omega}
 \omega(x)= -a \delta_0(x) - \frac{g}{|x|}.
\eq
Here $\delta_0$ is the Dirac-delta function at $0$, $a>0$ is the strength of the attractive interaction and $g>0$ is the gravitational constant, which will be set $=1$ for simplicity (our results are valid for all $g>0$). The energy of the condensate is described by the Gross-Pitaevskii functional 
$$
\cE_a(u)= \int_{\R^2} |\nabla u(x)|^2 \d x + \int_{\R^2} V(x) |u(x)|^2 \d x - \frac{a}{2} \int_{\R^2} |u(x)|^4 \d x -\iint \frac{|u(x)|^2|u(y)|^2}{|x-y|} \d x \d y.
$$
We are interested in the existence and properties of minimizers of the minimization problem
\bq \label{eq:GP}
E(a)=\inf\Big\{\cE_a(u) \,|\, u\in H^1(\R^2), \|u\|_{L^2}=1\Big\}.
\eq
Note that by the diamagnetic inequality $|\nabla u|\ge |\nabla |u||$, we can always restrict the consideration to the case $u \ge 0$. 
\medskip

For the systems of small scales, gravity is often omitted as it is normally much weaker than other forces. However, in the context of ultra-cold gas, the self-gravitating effect has gained increasing interest in physics. In particular, it is crucially relevant to the study of an analog of a black hole in a Bose-Einstein condensate, see e.g. \cite{B1,B2,B3,B4}. In the present paper, we will explain some interesting effects  of gravity in the instability of the condensate. 

By a simple scaling argument, we can see that $E(a)=-\infty$, if and only if $a\ge a^*$, where $a^*>0$ the optimal constant in the Gagliardo-Nirenberg inequality:
\bq 
\label{eq:GN} 
\int_{\R^2} |\nabla u(x)|^2 \d x \ge \frac{a^*}{2} \int_{\R^2} |u(x)|^4  \d x, \quad \forall u\in H^1(\R^2), \|u\|_{L^2}=1.
\eq
Indeed, it is well-known that (see e.g. \cite{GidNiNir-81,Weinstein-83,MclSer-87})
\bq \label{eq:a*-Q}
a^*=\int_{\R^2} |Q|^2 
\eq
where $Q$ is the unique positive, radial solution to the nonlinear Schr\"odinger equation \bq \label{eq:Q}
-\Delta Q + Q - Q^3 =0, \quad Q\in H^1(\R^2).
\eq 
Moreover, the normalized function $Q_0=Q/\|Q\|_{L^2}$ is the unique (up to translations and dilations) optimizer for the interpolation inequality \eqref{eq:GN}. Indeed,
\bq \label{eq:Q0}
1=\int_{\R^2} |Q_0|^2 = \int_{\R^2}|\nabla Q_0|^2 = \frac{a^*}{2} \int_{\R^2}|Q_0|^4.
\eq

In \cite{GuoSei-14}, Guo and Seiringer studied the collapse phenomenon of the Bose-Einstein without the self-gravitating effect (i.e. $g=0$ in \eqref{eq:omega}). They proved that with trapping potentials like $V(x)=|x|^p$, $p>0$, the Gross-Pitaevskii minimizer always exist when $a<a^*$ and they blow-up (possibly up to translations and dilations) to $Q_0$ as $a\uparrow a^*$. More precisely, if $u_a$ is a minimizer for $E(a)$, then  
\bq \label{eq:GS}
(a^*-a)^{\frac{1}{p+2}} u_a \Big(x (a^*-a)^{\frac{1}{p+2}}\Big) \to \beta Q_0(\beta x)
\eq
strongly in $L^2(\R^2)$, where
$$
\beta =  \left( \frac{p}{2} \int_{\R^2} |x|^p |Q(x)|^2 \d x \right)^{\frac{1}{p+2}} .
$$
The result in \cite{GuoSei-14}  has been extended to other kinds of external potentials, e.g. ring-shaped potentials \cite{GuoZenZho-16}, periodic potentials \cite{WanZha-16}, and Newton-like potentials \cite{Phan-17a}. 

\medskip

In the present paper, we will consider the existence and blow-up property of the Gross-Pitaevskii minimizers  in the case of having long-range self-gravitating interaction. It turns out that the self-gravitating interaction leads to interesting effects. For example, if the external potential $V$ is not singular enough, then the self-gravitating interaction is the main cause of the instability and the details of the blow-up phenomenon are more or less irrelevant to $V$. This situation is very different from  the case without gravity studied in \cite{GuoSei-14,GuoZenZho-16,WanZha-16,Phan-17a}. The precise form of our results will be provided in the next section.

\section{Main results}

Our first result is

\begin{theorem}[Existence] \label{thm:main1} Let $V:\R^2\to \R$ satisfy one of the following three conditions:
\begin{itemize}
\item[{\rm (V1)}] {\rm (Trapping potentials)} $V\ge 0$, $V\in L^{1}_{\rm loc}(\R^2)$ and $V(x)\to \infty$ as $|x|\to \infty$;

\item[{\rm (V2)}] {\rm (Periodic potentials)} $V\in C(\R^2)$ and $V(x+z)=V(x)$ for all $x\in \R^2$, $z\in \mathbb{Z}^2$;

\item[{\rm (V3)}] {\rm (Attractive potentials)} $V\le 0$ and $V\in L^p(\R^2)+L^q(\R^2)$ with $p,q\in (1,\infty)$.
\end{itemize}
Then there exists a constant $a_*<a^*$ such that $E(a)$ in \eqref{eq:GP} has a minimizer for all $a\in (a_*,a^*)$. We can choose $a_*=0$ in cases (V1) and  (V3). Moreover, $E(a)=-\infty$ for all $a\ge a^*$. 
\end{theorem}

Except the case of trapping potentials, the proof of the existence is non-trivial. Even in the case $V\equiv 0$, the Gross-Pitaevskii functional is translation-invariant and some mass may escape to infinity, leading to the lack of compactness. The existence result will be proved by the concentration-compactness method of Lions \cite{Lions-84,Lions-84b}. 

\medskip

Now we turn to the blow-up behavior of minimizers when $a\uparrow a^*$. First, we consider the case when the negative part of the external potential $V$ has no singular point, or it has some singular points but the singularity is weak. More precisely, we will assume
\bq \label{eq:weakV}
V\in L^1_{\rm loc}(\R^2), \quad V(x) \ge -C \sum_{j\in J} \frac{1}{|x-z_j|^{p}}, \quad 0<p<1
\eq
for a finite set $\{z_j\}_{j\in J}\subset \R^2$. We have

\begin{theorem}[Blow-up for weakly singular potentials] \label{thm:main2} Assume \eqref{eq:weakV}. Then  
\bq \label{eq:Ea-weak}
\limsup_{a\uparrow a^*} E(a)(a^*-a) = -\frac{a^*}{4}\Big( \iint \frac{|Q_0(x)|^2 |Q_0(y)|^2}{|x-y|} \d x \d y\Big)^2.
\eq 
Moreover, let $a_n\uparrow a^*$ and let $u_n \ge 0$ be an approximate minimizer for $E(a_n)$, i.e. $\cE_{a_n}(u_n)/E(a_n)\to 1$. Then there exist a subsequence $u_{n_k}$ and a sequence $\{x_k\}\subset \R^2$ such that
\bq \label{eq:weakV-cVu}
\lim_{k\to \infty} (a^*-a_{n_k}) u_{n_k}\big(x_k+ (a^*-a_{n_k})x\big)= \beta Q_0(\beta x) 
\eq
strongly in $H^1(\R^2)$ where
$$
\beta=\frac{a^*}{2} \iint \frac{|Q_0(x)|^2 |Q_0(y)|^2}{|x-y|} \d x \d y.
$$
Finally, if $V(x)=|x|^q$ for $q>0$,  or $V(x)=-|x|^{-q}$ for $0<q<1$, and if $u_n$ is a minimizer for $E(a_n)$ (which exists by Theorem \ref{thm:main1}), then the convergence \eqref{eq:weakV-cVu} holds true for the whole sequence and for $x_k=0$, i.e. 
\bq \label{eq:limimim}  \lim_{n\to \infty} (a^*-a_{n}) u_{n}\big( (a^*-a_{n})x\big)=\beta Q_0(\beta x). 
\eq
\end{theorem}

We observe that in Theorem \ref{thm:main2} the details of the blow-up phenomenon is essentially irrelevant to $V$. 
This is an interesting effect of the self-gravitating interaction. In the case without gravity studied in \cite{GuoSei-14,GuoZenZho-16,WanZha-16,Phan-17a}, the blow-up behavior depends crucially on the local behavior of $V$ around its minimizers/singular points, which can be seen from \eqref{eq:GS}. Heuristically, if the condensate shrinks with a length scale $\eps\to 0$, then the self-gravitating interaction is of order $\eps^{-1}$, while the external potential is of order at most $\eps^{-p}$ (since $V$ is not singular than $|x|^{-p}$).  Therefore, the contribution of the  external potential can be ignored to the leading order. 

\medskip

Now we come to the case when the external potential is more singular. We will assume 
\bq \label{eq:AS-V}
V(x)= - h(x) \sum_{j=1}^J  |x-z_j|^{- p_j}
\eq
with a finite set $\{z_j\}_{j\in J}\subset \R^2$, with
$$0<p_j<2, \quad p=\max_{j\in J} p_j\ge 1$$
and with
$$h\in C(\R^2), \quad C\ge  h \ge 0,\quad h_0=\max\{h(z_j):p_j=p\}>0.$$
Let us denote the set
$$
\mathcal{Z}=\{z_j: p_j=p, h_j=h_0\}
$$
which contains the most singular points of $V$. We have 

\begin{theorem}[Blow-up for strongly singular potentials] \label{thm:main3} Assume \eqref{eq:AS-V}. Then 
\begin{align} \label{eq:Ea-asym}
 \lim_{a\uparrow a^*} E(a) (a^*-a)^{\frac{p}{2-p}} = \frac{\beta^2}{a^*}-\beta^p A
 \end{align}
where
\begin{equation} \label{eq:betaAAA}
A=
\begin{cases} 
\, h_0 \int_{\R^2}  \frac{|Q_0(x)|^2}{|x|^{p}} \d x + \iint \frac{|Q_0(x)|^2 |Q_0(y)|^2}{|x-y|} \d x \d y , &\quad \text{if}\,\, p=1\\
\,h_0 \int_{\R^2}  \frac{|Q_0(x)|^2}{|x|^{p}} \d x. &\quad \text{if}\,\, p>1 \\
\end{cases}
\end{equation}
and
$$
\beta = \left(\frac{pa^*A}{2}\right)^{\frac{1}{p-2}}.
$$
Moreover,  let $a_n\uparrow a^*$ and let $u_n \ge 0$ be an approximate minimizer for $E(a_n)$, i.e. $\cE_{a_n}(u_n)/E(a_n)\to 1$. Then there exist a subsequence $u_{n_k}$ and a point $x_0\in \cal{Z}$ such that 
\bq \label{eq:cv-p>=1}
\lim_{n\to \infty} (a^*-a_{n_k})^{\frac{1}{2-p}} u_{a_{n_k}} \Big( x_{0} + x(a^*-a_n)^{\frac{1}{2-p}} \Big) = \beta Q_0(\beta x)
\eq
strongly in $H^1(\R^2)$. Finally, if $\mathcal{Z}$ has a unique element, then the convergence \eqref{eq:cv-p>=1} holds for the whole sequence $\{u_n\}$.  
\end{theorem}

In contrast of Theorem \ref{thm:main2}, Theorem \ref{thm:main3} says that if the external potential $V$ is singular enough, then its local behavior close its singular points determines the details of the blow-up profile. In particular, when $p>1$,  the convergences \eqref{eq:Ea-asym} and \eqref{eq:cv-p>=1} are similar to the results in \cite{Phan-17a} when there is no gravity term included in the energy functional. This means that in this case, the effect of the self-gravitating interaction is negligible to the leading order. 

\medskip

Note that the proof of the blow-up result \eqref{eq:GS} in \cite{GuoSei-14} is based on the analysis of the Euler-Lagrange equation associated to the minimizers. This approach has been used also in follow-up papers \cite{GuoZenZho-16,WanZha-16}. In the present paper, we will use another approach, which has the advantage that we can treat approximate minimizers as well (in principle there is no Euler-Lagrange equation for an approximate minimizer). 

\medskip

More precisely, as in \cite{Phan-17a}, we will prove the blow-up results by the energy method. First, we prove that approximate minimizers must be an optimizing sequence for the Gagliardo-Nirenberg inequality \eqref{eq:GN}. Then by the concentration-compactness argument, up to subsequences, translations and dilations, this sequence converges to an optimizer for \eqref{eq:GN}, which is of the form $bQ_0(bx+x_0)$ for some $b>0$ and $x_0\in \R^2$. Then we determine $b$ and $x_0$ by matching the asymptotic formula for $E(a)$. 
 
\medskip

{\bf Heuristic argument.}  Now let us explain the heuristic ideas of our analysis of the blow-up phenomenon. For simplicity, consider the case $V(x)=-|x|^{-p}$ with $0<p<2$ and the energy functional becomes
$$
\cE_a(u)= \int_{\R^2} |\nabla u(x)|^2 \d x - \int_{\R^2} \frac{ |u(x)|^2}{|x|^p} \d x - \frac{a}{2} \int_{\R^2} |u(x)|^4 \d x -\iint \frac{|u(x)|^2|u(y)|^2}{|x-y|} \d x \d y.
$$
If the minimizer $u_a$ of $\cE_a(u)$ converges to $Q_0$ under the length-scaling $\ell$, i.e. $u_a(x)\approx \ell Q_0(\ell x)$, then 
\begin{align} \label{eq:heu}
\cE_a(u_a) &\approx \cE_a(\ell Q_0 (\ell \cdot)) = \ell^2  \int_{\R^2}|\nabla Q_0(x)|^2 \d x  -\ell^p \int_{\R^2} \frac{|Q_0(x)|^2}{|x|^p} \d x \nn\\
&\qquad - \frac{a \ell^2 }{2}\int_{\R^2} |Q_0(x)|^4\d x - \ell \iint_{\R^2 \times \R^2} \frac{|Q_0(x)|^2 |Q_0(y)|^2}{|x-y|} \d x \d y\nn\\
&= \ell^2 \Big(1- \frac{a}{a^*}\Big) -\ell^p \int_{\R^2} \frac{|Q_0(x)|^2}{|x|^p} \d x - \ell \iint_{\R^2 \times \R^2} \frac{|Q_0(x)|^2 |Q_0(y)|^2}{|x-y|} \d x \d y.
\end{align}
We want to choose $\ell$ to minimize  the ride side of \eqref{eq:heu} (since $u_a$ minimizes $\cE_a(u)$). It is not hard to guess that when $a\uparrow a^*$, then $\ell \to \infty$ and the exact behavior of $\ell$ depends on $p$ as follows.
\begin{itemize}
\item If $p<1$ then $\ell^p\ll \ell$, and the term of order $\ell^p$ does not contribute to the leading order. The value of $\ell$ is essentially  determined by minimizing
$$
\ell^2 \Big(1- \frac{a}{a^*}\Big) -  \ell \iint_{\R^2 \times \R^2} \frac{|Q_0(x)|^2 |Q_0(y)|^2}{|x-y|} \d x \d y,
$$
i.e.
$$
\ell \approx \frac{a^*}{2(a^*-a)} \iint_{\R^2 \times \R^2} \frac{|Q_0(x)|^2 |Q_0(y)|^2}{|x-y|} \d x \d y.
$$
This is the situation covered in Theorem \ref{thm:main2}. 

\item If $p>1$, then $\ell^p\gg \ell$, and the value of $\ell$ is essentially determined by minimizing
$$
\ell^2 \Big(1- \frac{a}{a^*}\Big) -  \ell^p \int_{\R^2} \frac{|Q_0(x)|^2}{|x|^p} \d x
$$
i.e.
$$
\ell \approx \left[ \frac{a^* p}{2(a^*-a)} \int_{\R^2}\frac{|Q_0(x)|^2}{|x|^p} \d x \right]^{\frac{1}{2-p}}.
$$
On the other hand, if  $p=1$, then $\ell^p=\ell$ and the value of $\ell$ is essentially determined by minimizing
$$
\ell^2 \Big(1- \frac{a}{a^*}\Big) -  \ell \Big[ \int_{\R^2} \frac{|Q_0(x)|^2}{|x|} \d x + \iint_{\R^2 \times \R^2} \frac{|Q_0(x)|^2 |Q_0(y)|^2}{|x-y|} \d x \d y \Big]
$$
i.e.
$$
\ell \approx \frac{a^*}{2(a^*-a)} \Big[ \int_{\R^2} \frac{|Q_0(x)|^2}{|x|} \d x  + \int_{\R^2 \times \R^2} \frac{|Q_0(x)|^2 |Q_0(y)|^2}{|x-y|} \d x \d y \Big]
$$
These situations are covered in Theorem \ref{thm:main3}. \end{itemize}

\medskip

{\bf Varying the gravitation constant.} As mentioned, here we consider the interaction of the form \eqref{eq:omega}, with the gravitation constant $g=1$ for simplicity. Clearly, our results hold for any constant $g>0$ (independent of $a$), up to easy modifications. It might be interesting to ask what happens when $g\to 0$ or $g\to \infty$, at the same time as $a\to a^*$. By following the above heuristic discussion, if $V(x)=-|x|^{-p}$ with $0<p<2$ and $u_a(x)\approx \ell Q_0(\ell x)$, then 
\begin{align} \label{eq:heu-2}
\cE_a(u_a) \approx  \ell^2 \Big(1- \frac{a}{a^*}\Big) -\ell^p \int_{\R^2} \frac{|Q_0(x)|^2}{|x|^p} \d x - \ell g \iint_{\R^2 \times \R^2} \frac{|Q_0(x)|^2 |Q_0(y)|^2}{|x-y|} \d x \d y.
\end{align}
Next we minimize the right side of \eqref{eq:heu-2}. If $\ell g \ll \ell^p$, namely $g\ll \ell^{p-1}$, then we can ignore the gravitation term and the optimal value of $\ell$ is $\sim (a^*-a)^{-\frac{1}{2-p}}$. This suggests that the threshold for the gravitation effect to be visible in the blow-up profile is $g\sim (a^*-a)^{\frac{1-p}{2-p}}$.  More precisely, we can expect the following:
\begin{itemize}
\item If $g\gg (a^*-a)^{\frac{1-p}{2-p}}$, then the contribution of the $\ell^p$-term can be ignored to the leading order, and 
$$
\ell \approx \frac{a^*g}{2(a^*-a)} \iint_{\R^2 \times \R^2} \frac{|Q_0(x)|^2 |Q_0(y)|^2}{|x-y|} \d x \d y.
$$
Thus the blow-up profile is determined completely by the gravitation term if either $p<1$ and $g\to 0$ slowly enough, or $p\ge 1$ and $g\to +\infty$ fast enough. 

\item  If $g\ll (a^*-a)^{\frac{1-p}{2-p}}$, then the contribution of the $\ell g$-term can be ignored to the leading order, and 
$$
\ell \approx \left[ \frac{a^* p}{2(a^*-a)} \int_{\R^2}\frac{|Q_0(x)|^2}{|x|^p} \d x \right]^{\frac{1}{2-p}}.
$$
Thus the blow-up profile is determined only by the attractive potential term (i.e. no gravitational effect) if either $p\le 1$ and $g\to 0$ fast enough, or  $p> 1$ and $g\to +\infty$ slow enough. 

\item If $g\sim (a^*-a)^{\frac{1-p}{2-p}}$, then both potential term and gravitation term enter the determination of the blow-up profile. 

\end{itemize}

Although our representation will focus on the case when $g$ is independent of $a$, as stated in Theorem \ref{thm:main2} and Theorem \ref{thm:main3}, the interested reader can prove the above assertions when $g$ is dependent on $a$ by following our analysis below. 

\medskip

{\bf Organization of the paper.} In the rest, we will prove Theorems \ref{thm:main1}, \ref{thm:main2}, \ref{thm:main3}  in Sections  \ref{sec:main1}, \ref{sec:main2}, \ref{sec:main3}, respectively. Also, for the reader's convenience, we recall in Appendix \ref{sec:app} the Concentration-Compactness Lemma and a standard result on the compactness of the optimizing sequences for the Gagliardo-Nirenberg inequality \eqref{eq:GN}, which are useful in our proof. 

\medskip

{\bf Acknowledgement.} I thank the referee for the interesting suggestion of considering the case when the gravitational constant $g$ depends on $a$, leading to an improvement on the representation of the paper. 

\section{Existence}\label{sec:main1}

In this section we prove Theorem \ref{thm:main1}. We will always denote by $C$ a universal, large constant.

We start with a preliminary result, which is the upper bound in \eqref{eq:Ea-weak}.

\begin{lemma} \label{lem:1} For all $V\in L^1_{\rm loc}(\R^2)$, we have the upper bound
\bq 
\label{eq:lima}
\limsup_{a\uparrow a^*} E(a)(a^*-a) \le -\frac{a^*}{4}\Big( \iint \frac{|Q_0(x)|^2 |Q_0(y)|^2}{|x-y|} \d x \d y\Big)^2.
\eq 
\end{lemma}

\begin{proof} As in \cite{GuoSei-14} we use the trial function
$$
u_\ell(x)=A_\ell \varphi(x-x_0) \ell Q_0(\ell(x-x_0))
$$
where $x_0\in \R^2$, $0\le \varphi \in C_c^\infty (\R^2)$ with $\varphi(x)=1$ for $|x|\le 1$, and $A_\ell>0$ is a normalizing factor to ensure $\|u\|_{L^2}=1$. 

Using \eqref{eq:Q0} and the fact that both $Q_0$ and $|\nabla Q_0|$ are exponentially decay (see \cite[Proposition 4.1]{GidNiNir-81}), we have 

\begin{align*}
A_\ell&=\Big(\int_{\R^2} |\varphi(x)|^2 \ell^2 |Q_0(\ell x)|^2 \d x \Big)^{-1/2}= 1+O(\ell^{-\infty}),\\
\int |\nabla u_\ell|^2 &= \ell^2 \int |\nabla Q_0|^2 + O(\ell^{-\infty}),\\
\int |u_\ell|^4&=\ell^2 \int | Q_0|^4+ O(\ell^{-\infty})=\ell^2 \frac{2}{a^*} \int |\nabla Q_0|^2 + O(\ell^{-\infty})
\end{align*}
and
$$
\iint \frac{|u_\ell(x)|^2 |u_\ell(y)|^2}{|x-y|} \d x \d y = \ell \iint \frac{|Q_0(x)|^2 |Q_0(y)|^2}{|x-y|} \d x \d y + O(\ell^{-\infty}),
$$
where $O(\ell^{-\infty})$ means that this quantity converges to $0$ faster than $\ell^{-k}$ when $\ell\to \infty$ for all $k=1,2,...$ Moreover, since $x\mapsto V(x) |\varphi(x-x_0)|^2$ is integrable and $\ell^2 |Q_0(\ell(x-x_0))|^2$ converges weakly to Dirac-delta function at $x_0$ when $\ell\to \infty$, we have
\begin{align*}
\int_{\R^2} V |u|^2  = {|A_\ell|^2} \int_{\R^2} V(x) |\varphi(x-x_0)|^2  |Q_0(\ell(x-x_0))|^2 \ell^2 \d x \to V(x_0) 
\end{align*}
for a.e. $x_0\in \R^2$. Using \eqref{eq:Q0} we thus obtain  
\begin{align*}
E(a)\le \cE_a(u)= \ell^2 \Big( 1- \frac{a}{a^*} \Big) + V(x_0) - \ell \iint \frac{|Q_0(x)|^2 |Q_0(y)|^2}{|x-y|} \d x \d y + o(1)_{\ell\to \infty}.
\end{align*}
Choosing $\ell=\lambda (a^*-a)^{-1}$ with a constant $\lambda>0$ and take $a\uparrow a^*$, we obtain
\bq
\limsup_{a\uparrow a^*}E(a) (a^*-a)\le  \frac{\lambda^2}{a^*}- \lambda \iint \frac{|Q_0(x)|^2 |Q_0(y)|^2}{|x-y|} \d x \d y.
\eq
Choosing the optimal value
\bq
\lambda=\frac{a^*}{2} \iint \frac{|Q_0(x)|^2 |Q_0(y)|^2}{|x-y|} \d x \d y
\eq
leads to the desired result. 
\end{proof}

Next, we have the a simple lower bound for $\cE_a(u)$.

\begin{lemma} \label{lem:cEa-lbs}For all $u\in H^1(\R^2)$ with $\|u\|_{L^2}=1$, we have 
\bq \label{eq:cEu>-}
\cE_a(u) \ge \left(1-\frac{a}{a^*}-\eps \right) \int_{\R^2} |\nabla u|^2 + \int_{\R^2} V |u|^2 - \frac{C}{\eps}, \quad \forall \eps>0.
\eq
\end{lemma}

\begin{proof} Take arbitrarily $u\in H^1(\R^2)$ with $\|u\|_{L^2}=1$. By the Gagliardo--Nirenberg inequality \eqref{eq:GN}, we have
$$
\frac{a}{2}\int |u|^4 \le \frac{a}{a^*} \int |\nabla u|^2. 
$$
Moreover, by the Hardy-Littewood-Sobolev inequality \cite[Theorem 4.3]{LiLo-01},  H\"older's inequality and the Gagliardo--Nirenberg inequality \eqref{eq:GN} again we have
\begin{align} \label{eq:VVVV}
\iint_{\R^2\times \R^2} \frac{|u(x)|^2|u(y)|^2}{|x-y|} \d x \d y &\le C\|u\|_{L^{8/3}}^4 \le C\|u\|_{L^4}^2 \|u\|_{L^2}^2 \le \eps \int |\nabla u|^2 + \frac{C}{\eps} 
\end{align}
for all $\eps>0$. Combining these estimates, we obtain the desired lower bound. 
\end{proof}

Now we go to the proof of Theorem \ref{thm:main1}. From Lemma \ref{lem:1} and the fact that $a\mapsto E(a)$ is non-increasing, we deduce that $E(a)=-\infty$ for all $a\ge a^*$.

Next, consider $a<a^*$. We distinguish three cases when the external potential $V$ satisfies (V1), (V2) or (V3), respectively.    

\begin{lemma}[Trapping potentials] Assume that (V1) holds. Then $E(a)$ has a minimizer for all $a\in (0,a^*)$.
\end{lemma}

\begin{proof} From \eqref{eq:cEu>-} and the assumption $V\ge 0$, we have $E(a)>-\infty$. Moreover, if $\{u_n\}$ is a minimizing sequence for $E(a)$, then $\|u_n\|_{H^1}$ and $\int V|u_n|^2$ are bounded. By Sobolev's embedding, after passing to a subsequence if necessary, we can assume that $u_n$ converges to a function $u_0$ weakly in $H^1(\R^2)$ and pointwise.  

For every $R>0$, $u_n\to u_0$ strongly in $L^2(B(0,R))$ by Sobolev's embedding. Therefore, 

\begin{align*}
\int_{|x|\le R} |u_0|^2 &= \lim_{n\to \infty}\int_{|x|\le R} |u_n|^2 = 1 - \lim_{n\to \infty}\int_{|x|> R} |u_n|^2 \\
&\ge 1 - \Big(\inf_{|z|>R} V(z)\Big)^{-1}\limsup_{n\to \infty}\int_{|x|> R} V |u_n|^2 \ge 1- C \Big(\inf_{|z|>R} V(z)\Big)^{-1}.
\end{align*} 
Taking $R\to \infty$ and $V(x)\to \infty$ as $|x|\to \infty$, we obtain
$$
\int_{\R^2} |u_0|^2\ge \lim_{R\to \infty}\int_{|x|\le R} |u_0|^2 \ge 1.
$$
Since we have known $u_n\to u_0$ weakly in $H^1(\R^2)$, we can conclude that $u_n\to u_0$ strongly in $L^2(\R^2)$. By Sobolev's embedding again, $u_n\to u_0$ strongly in $L^p(\R^2)$ for all $p\in [2,\infty)$. Consequently, $\|u_0\|_{L^2}=1$, 
$$
\int |u_n|^4 \to \int |u_0|^4
$$
and, by the Hardy-Littewood-Sobolev inequality,
$$
\iint_{\R^2\times \R^2} \frac{|u_n(x)|^2|u_n(y)|^2}{|x-y|} \d x \d y \to \iint_{\R^2\times \R^2} \frac{|u_0(x)|^2|u_0(y)|^2}{|x-y|} \d x \d y.
$$
Moreover, by Fatou's lemma we have
$$
\int |\nabla u_n|^2 \ge \int |\nabla u_0|^2 + o(1)
$$
as $u_n\to u$ weakly in $H^1(\R^2)$ and
$$
\int V|u_n|^2 \ge \int V |u_0|^2 + o(1)
$$
as $u_n\to u_0$ pointwise. In summary, 
$$
\cE_a(u_n)\ge \cE_a(u_0)+o(1).
$$
Since $u_n$ is a minimizing sequence, we conclude that $E(a)\ge \cE_a(u_0)$, i.e. $u_0$ is a minimizer.
\end{proof}

\begin{lemma}[Periodic potentials] Assume that (V2) holds. Then there exists $a_*<a^*$ such that $E(a)$ has a minimizer for all $a\in (a_*,a^*)$. 
\end{lemma}

\begin{proof}
Since $V$ is continuous and periodic, it is uniformly bounded. From \eqref{eq:cEu>-}, we have $E(a)>-\infty$ for all $a\in (0,a^*)$. Moreover, by Lemma \ref{lem:1} we have $E(a)\to -\infty$ as $a\uparrow a^*$. Therefore, we can find $a_*<a^*$ such that for all $a\in (a_*,a^*)$ we have
\bq \label{eq:E<V}
E(a)<\inf V.
\eq

Now we prove $E(a)$ has minimizers for all $a\in (a_*,a^*)$. We will use the concentration-compactness method of Lions \cite{Lions-84,Lions-84b}. For the reader's convenience, we summary all we need in Lemma \ref{lem:concom} in  Appendix. 

Let $\{u_n\}$ be a minimizing sequence for $E(a)$.  From \eqref{eq:cEu>-}, $u_n$ is bounded in $H^1$. Hence, we can apply Concentration-Compactness Lemma \ref{lem:concom} to the sequence $\{u_n\}$.  Up to subsequences, one of the three cases in  Lemma \ref{lem:concom} must occur.

\medskip

{\bf No-vanishing.} Assume that the vanishing  case (ii) in Lemma \ref{lem:concom} occurs. Then  we have
$$
\int |u_n|^4 \to 0, \quad \iint_{\R^2\times \R^2} \frac{|u_n(x)|^2|u_n(y)|^2}{|x-y|} \d x \d y \to 0
$$
(in the latter we used the Hardy-Littewood-Sobolev inequality). Combining with the obvious lower bound
$$
\int V|u_n|^2 \ge \inf V
$$
we find that
\bq \label{eq:E>=V}
E(a)=\lim_{n\to \infty}\cE_a(u_n) \ge \inf V.  
\eq
However, this contradicts to \eqref{eq:E<V}. Thus the vanishing case cannot occur.

\medskip

{\bf No-dichotomy.} Assume that the dichotomy case (ii) in Lemma \ref{lem:concom} occurs, i.e. we can find two sequences $u_n^{(1)}$, $u_n^{(2)}$ such that 
\[ 
\left \{
\begin{aligned}
  & \int_{\R^2} |u_n^{(1)}|^2 \to  \lambda,  \quad \int_{\R^2} |u_n^{(2)}|^2 \to  1-\lambda,  \\
  & {\rm dist} ({\rm supp} (u_n^{(1)}), {\rm supp} (u_n^{(2)})) \to +\infty;\\
  & \|u_n - u_n^{(1)}-u_n^{(2)}\|_{L^p}\to  0  , \quad  \forall p\in [2,\infty); \\
  & \int_{\R^2} (|\nabla u_n|^2 - |\nabla u_n^{(1)}|^2 - |\nabla u_n^{(2)}|^2) \ge o(1).
  \end{aligned}
 \right.
\]
Since $\|u_n - u_n^{(1)}-u_n^{(2)}\|_{L^p}\to  0$ and ${\rm dist} ({\rm supp} (u_n^{(1)}), {\rm supp} (u_n^{(2)})) \to +\infty$ we have
$$
\int_{\R^2} |u_n|^4 =\int_{\R^2}  |u_n^{(1)}+u_n^{(2)}|^4 + o(1) = \int_{\R^2} \Big( |u_n^{(1)}|^4 + |u_n^{(2)}|^4) +o(1)
$$
and
\begin{align}
  &\iint_{\R^2\times \R^2} \frac{|u_n(x)|^2|u_n(y)|^2}{|x-y|} \d x \d y \\ \label{eq:HLS-cou}
 &= \iint_{\R^2\times \R^2} \frac{|u_n^{(1)}(x)+u_n^{(2)}(x)|^2|u_n^{(1)}(y)+u_n^{(2)}(y)|^2}{|x-y|} \d x \d y +o(1) \\
 &= \iint_{\R^2\times \R^2} \frac{(|u_n^{(1)}(x)|^2+|u_n^{(2)}(x)|^2) (|u_n^{(1)}(y)|^2+|u_n^{(2)}(y)|^2)}{|x-y|} \d x \d y + o(1)\nn\\\label{eq:HLS-coub}
 &= \iint_{\R^2\times \R^2} \frac{|u_n^{(1)}(x)|^2 |u_n^{(1)}(y)|^2+|u_n^{(2)}(x)|^2 |u_n^{(2)}(y)|^2}{|x-y|} \d x \d y + o(1)
\end{align}
Here we have used the Hardy-Littewood-Sobolev inequality in \eqref{eq:HLS-cou} and used the decay of Newton potential, i.e. $|x-y|^{-1}\to 0$ as $|x-y|\to \infty$, to remove the cross-term 
$$\iint \frac{ |u_n^{(1)}(x)|^2 |u_n^{(2)}(y)|^2+|u_n^{(2)}(x)|^2 |u_n^{(1)}(y)|^2}{|x-y|} \d x \d y$$
in \eqref{eq:HLS-coub}. 

 Moreover, since $V$ is bounded we get
\bq
\int_{\R^2} V\left(|u_n|^2-|u_1^{(1)}|^2-|u_n^{(2)}|^2\right)= \int_{\R^2} V\left(|u_n|^2-|u_1^{(1)}+u_n^{(2)}|^2\right)\to 0
\eq
In summary, we have the energy decomposition
\bq \label{eq:Eun-un12}
\cE_a(u_n) \ge \cE_a(u_n^{(1)}) + \cE_a(u_n^{(2)})+o(1).
\eq

Next, we use $u_n^{(1)}/\|u_n^{(1)}\|_{L^2}$ as a trial state for $E(a)$. By the variational principle, we have
\begin{align} \label{eq:decom-Eaaa}
E(a) \le \cE\Big(\frac{u_n^{(1)}}{\|u_n^{(1)}\|_{L^2}}\Big) &= \lambda^{-1} \Big( \int |\nabla u_n^{(1)}|^2 + \int V|u_n^{(1)}|^2 \Big)\nn\\
& - \lambda^{-2} \Big( \frac{a}{2}\int |u_n^{(1)}|^4 + \iint \frac{|u_n^{(1)}(x)|^2 |u_n^{(1)}(y)|^2}{|x-y|} \d x \d y\Big)+o(1).
\end{align}
In the latter equality, we have used $\|u_n^{(1)}\|_{L^2}^2\to \lambda$ (and $u_n^{(1)}$ is bounded in $H^1$). The above inequality can be rewritten as 
\bq \label{eq:cEun1}
\lambda^2 E(a) + (1-\lambda) \Big( \int |\nabla u_n^{(1)}|^2 + \int V|u_n^{(1)}|^2 \Big) \le \cE_a(u_n^{(1)})+o(1).
\eq
By ignoring the kinetic energy on the left side and the obvious bound
$$
\int V|u_n^{(1)}|^2 \ge (\inf V) \int  |u_n^{(1)}|^2 = \lambda (\inf V)+o(1)
$$
we find that
$$
\lambda^2 E(a) + \lambda (1-\lambda) \inf V  \le \cE_a(u_n^{(1)}) +o(1).
$$
Similarly, since $\|u_n^{(1)}\|_{L^2}^2\to 1-\lambda$ we get
$$
(1-\lambda)^2 E(a) +\lambda (1-\lambda) \inf V  \le \cE_a(u_n^{(2)})+o(1).
$$
Summing the latter inequalities gives 
$$
\cE_a(u_n^{(1)}) + \cE_a(u_n^{(2)}) \ge (\lambda^2+(1-\lambda)^2) E(a)+ 2 \lambda (1-\lambda) \inf V  + o(1).
$$
Inserting this into \eqref{eq:Eun-un12} and using $\cE_a(u_n)\to E(a)$ we arrive at
$$
E(a) \ge (\lambda^2+(1-\lambda)^2) E(a)+ 2 \lambda (1-\lambda) \inf V 
$$
which is equivalent to
$$
E(a)\ge \inf V
$$
because $\lambda \in (0,1)$. But again, it is a contradiction to \eqref{eq:E<V}. Thus the dichotomy case cannot occur.

\medskip

{\bf Compactness.} Now we can conclude that the compactness case (i) in Lemma \ref{lem:concom} must occur, i.e. there exists a sequence $\{z_n\}\subset \R^2$ such that $\widetilde u_n= u_n(.+z_n)$ converges to some $u_0$ weakly in $H^1(\R^2)$ and strongly in $L^p(\R^2)$ for all $p\in [2,\infty)$. Since the Lebesgue measure is translation-invariant, we have $\|u_0\|_{L^2}=1$,
$$
\int_{\R^2} |u_n|^4  = \int_{\R^2} |u_n(x+z_n)|^4 \d x =  \int_{\R^2} |u_0|^4 + o(1)_{n\to \infty},
$$

$$
\int_{\R^2} V |u_n|^2 = \int_{\R^2} V(x+z_n) |u_n(x+z_n)|^2 \d x = \int_{\R^2} V(x+z_n) |u_0(x)|^2 \d x + o(1)_{n\to \infty}.
$$
by the boundedness of $V$, and 

\begin{align*}
\iint_{\R^2} \frac{|u_n(x)|^2 |u_n(y)|^2}{|x-y|}\d x \d y &=\iint_{\R^2} \frac{|u_n(x+z_n)|^2 |u_n(y+z_n)|^2}{|x-y|}\d x \d y\\
&=\iint_{\R^2} \frac{|u_0(x)|^2 |u_0(y)|^2}{|x-y|}\d x \d y + o(1)_{n\to \infty}
\end{align*}
by the Hardy-Littewood-Sobolev inequality.

Moreover,  since $\nabla u_n\to \nabla u$ weakly in $L^2(\R^2)$ we have
\begin{align*}
\int_{\R^2} |\nabla u_n|^2  = \int_{\R^2} |\nabla u_n(x+z_n)|^2 \d x \ge  \int_{\R^2} |\nabla u_0|^2 + o(1)_{n\to \infty} .
\end{align*} 
In summary, since $u_n$ is a minimizing sequence, we conclude that
\begin{align} \label{eq:cEa>=cEau0}
\cE_a(u_n) &\ge  \int_{\R^2} |\nabla u_0|^2 + \int  V(\cdot + z_n)|u_0|^2 \nn\\
&- \frac{a}{2} \int |u_0(x)|^4 - \iint_{\R^2} \frac{|u_0(x)|^2 |u_0(y)|^2}{|x-y|}\d x \d y   + o(1)_{n\to \infty}.
\end{align}

To finish, we use the periodicity of $V$. We can write  
$$
z_n = y_n + z\quad \text{with}\quad y_n\in [0,1]^2, z\in \mathbb{Z}^2.
$$
Since $y_n$ is bounded, up to a subsequence, we can assume that $y_n\to y_0$ in $\R^2$. Thus by the periodicity of $V$ and the Lebesgue Dominated Convergence, we have
$$
\int  V(x+ z_n)|u_0(x)|^2 \d x =  \int  V(x+ y_n)|u_0(x)|^2 \d x \to  \int  V(x+ y_0)|u_0(x)|^2 \d x.
$$
Thus \eqref{eq:cEa>=cEau0} reduces to
\begin{align*}
\cE_a(u_n) &\ge  \int_{\R^2} |\nabla u_0|^2 + \int  V(\cdot + y_0)|u_0|^2 \nn\\
&- \frac{a}{2} \int |u_0(x)|^4 - \iint_{\R^2} \frac{|u_0(x)|^2 |u_0(y)|^2}{|x-y|}\d x \d y   + o(1)_{n\to \infty}\\
&= \cE_a(u_0(\cdot -y_0))+o(1)_{n\to \infty}.
\end{align*}
Since $u_n$ is a minimizing sequence, we conclude that $u_0(\cdot -y_0)$ is a minimizer for $E(a)$. 
\end{proof}

\begin{lemma}[Attractive potentials] \label{lem:V3}Assume that (V3) holds. Then $E(a)$ has a minimizer for all $a\in (0,a^*)$. 
\end{lemma}
\begin{proof} First, if $V\equiv 0$, then we can follow the proof of the periodic case and use the strict inequality $E(a)<0$ (i.e. \eqref{eq:E<V} holds true with $a_*=0$). To prove $E(a)<0$, we can simply take the trial function 
$$u_\ell( \ell x)=\ell Q_0(\ell x)$$
and use the variational principle
$$
E(a)\le \cE(u_\ell)= \ell^2 \Big( 1-\frac{a}{a^*}\Big) -\ell \iint \frac{|Q_0(x)|^2|Q_0(y)|^2}{|x-y|} \d x \d y
$$
with $\ell>0$ sufficiently small. 

Now we come to the case when $V\not \equiv 0$. Since $V\in L^p(\R^2)+L^q(\R^2)$ with $1<p<q<\infty$, by Sobolev's inequality, we have the lower bound (see \cite[Section 11.3]{LiLo-01})
\bq \label{eq:So}
\eps \int |\nabla u|^2  + \int  V|u|^2 \ge -C_\eps , \quad \forall \eps>0, \quad \forall u\in  H^1(\R^2), \|u\|_{L^2}=1.
\eq
Thus \eqref{eq:cEu>-} reduces to 
$$
\cE_a(u) \ge \left(1-\frac{a}{a^*}-\eps \right) \int_{\R^2} |\nabla u|^2 - C_\eps, \quad \forall \eps>0, \quad \forall u\in H^1(\R^2), \|u\|_{L^2}=1.
$$
This ensures that $E(a)>-\infty$. Moreover, if $\{u_n\}$ is a minimizing sequence for $E(a)$, then $u_n$ is bounded in $H^1(\R^2)$. By Sobolev's embedding, after passing to a subsequence if necessary, we can assume that $u_n$ converges to a function $u_0$ weakly in $H^1(\R^2)$ and pointwise.  

\medskip

{\bf Energy decomposition.} Now following the concentration-compactness argument, we will show that
\bq \label{eq:cEun-cEu0-cEu}
\cE_a(u_n)\ge \cE_a(u_0)+ \cE^\infty_a(u_n-u_0)+o(1)
\eq
with $\cE_a^\infty(\varphi)$ the energy functional without the external potential, i.e.
$$
\cE_a^\infty(\varphi)=\int |\nabla \varphi|^2 -\frac{a}{2}\int |\varphi|^4 -\iint \frac{|\varphi(x)|^2 |\varphi(y)|^2}{|x-y|} \d x \d y.
$$

Indeed, since $u_n\to u_0$ weakly in $H^1(\R^2)$, we have  $\nabla u_n\to \nabla u$ weakly in $L^2$ and hence 
\begin{align*}
\|\nabla u_n\|_{L^2}^2 &=  \|\nabla u_0\|_{L^2}^2 + \|\nabla (u_n-u_0)\|_{L^2}^2 + 2 \langle \nabla u_0, \nabla (u_n-u_0)\rangle \\
&=  \|\nabla u_0\|_{L^2}^2 + \|\nabla (u_n-u_0)\|_{L^2}^2 +o(1).
\end{align*}
Moreover, since $u_n\to u_0$ weakly in $H^1(\R^2)$ and $V\in L^p(\R^2)+L^q(\R^2)$ with $1<p<q<\infty$ we have (see \cite[Theorem 11.4]{LiLo-01})
$$
\int V|u_n|^2 \to \int V|u_0|^2.
$$
Moreover, since $u_n\to u_0$ pointwise, we have 
$$
\int |u_n|^4= \int |u_0|^4 + \int |u_n-u_0|^4 + o(1)
$$
by Brezis-Lieb's refinement of Fatou's lemma \cite{BreLie-83} 
\begin{align*}
&\iint \frac{|u_n(x)|^2 |u_n(y)|^2}{|x-y|} \d x \d y \\
&=\iint \frac{|u_0(x)|^2 |u_0(y)|^2}{|x-y|} \d x \d y + \iint \frac{|u_n(x)-u_0(x)|^2 |u_n(y)-u_0(y)|^2}{|x-y|} \d x \d y +o(1)
\end{align*}
by a nonlocal version of the Brezis-Lieb lemma in \cite[Lemma 2.2]{BFV-14}.  Thus \eqref{eq:cEun-cEu0-cEu} holds true. 

\medskip

{\bf No-vanishing.} Next, we show that $u_0\not\equiv 0$. Assume by contradiction that $u_0\equiv 0$. Then \eqref{eq:cEun-cEu0-cEu} implies that
\bq \label{eq:E=Einf}
E(a)= \lim_{n\to \infty}\cE_a(u_n) \ge \lim_{n\to \infty}\cE_a^\infty(u_n) \ge E^\infty(a)
\eq
where
$$
E^\infty(a)= \inf\{\cE^\infty_a(\varphi)\,|\,\varphi\in H^1(\R^2), \|\varphi\|_{L^2}=1\}.
$$

On the other hand, we have proved that $E^\infty(a)$ has a minimizer $\varphi_0$ (this is the case when the external potential vanishes). We can assume $\varphi_0\ge 0$ by the diamagnetic inequality. By a standard variational argument, $\varphi_0$ solves the Euler-Lagrange equation 
\bq \label{eq:EL0}
-\Delta \varphi_0 - a\varphi_0^3 - 2(|\varphi_0|^2*|x|^{-1})\varphi_0 =\mu \varphi_0
\eq
for a constant $\mu\in \R$ (the Lagrange multiplier). Consequently, we have 
$(-\Delta+|\mu|)\varphi_0 \ge 0$, and hence $\varphi_0>0$ by \cite[Theorem 9.9]{LiLo-01}. 

From  the facts that $V\le 0$, $V\not\equiv 0$ and $\varphi_0>0$, we have
$$
\cE_a(\varphi_0)-\cE^\infty_a(\varphi_0) = \int V|\varphi_0|^2<0.
$$
Therefore, by the variational principle,
\bq \label{eq:E<Einf}
E(a) < E^\infty(a). 
\eq
Thus \eqref{eq:E=Einf} cannot occur, i.e. we must have that $u_0\not\equiv 0$.

\medskip

{\bf Compactness.} It remains to show that $\|u_0\|_{L^2}=1$. We assume by contradiction that $\|u_0\|_{L^2}^2=\lambda\in (0,1)$. Then similarly to \eqref{eq:decom-Eaaa} we have 
\begin{align*}
E(a) \le \cE\Big(\frac{u_0}{\|u_0\|_{L^2}}\Big) &= \lambda^{-1} \Big( \int |\nabla u_0|^2 + \int V|u_0|^2 \Big)\nn\\
& - \lambda^{-2} \Big( \frac{a}{2}\int |u_0|^4 + \frac{1}{2} \iint \frac{|u_0(x)|^2 |u_0(y)|^2}{|x-y|} \d x \d y\Big)\nn\\
&\le \lambda^{-1} \cE_a(u_0).\nn
\end{align*}
Thus
\begin{align} \label{eq:Vfn1}
\cE_a(u_0)\ge \lambda E(a).
\end{align}
Similarly, using $\|u_n-u_0\|^2\to 1-\lambda$ we get
\begin{align} \label{eq:Vfn2}
\cE_a^\infty(u_n-u_0)\ge (1-\lambda) E^\infty(a) + o(1).
\end{align}
Inserting \eqref{eq:Vfn1} and \eqref{eq:Vfn2} into \eqref{eq:cEun-cEu0-cEu} and using $\cE_a(u_n)\to E(a)$, we get
$$
E(a)\ge  \lambda E(a) + (1-\lambda) E^\infty(a). 
$$
which is equivalent to $E(a)=E^\infty(a)$ as $\lambda\in (0,1)$. However, it contradicts to \eqref{eq:E<Einf}. 

Thus we conclude that $\|u_0\|_{L^2}=1$. Now \eqref{eq:Vfn2} becomes $\cE_a^\infty(u_n-u_0)\ge o(1)$, which, together with   \eqref{eq:cEun-cEu0-cEu}, implies that
$$
E(a)=\lim_{n\to \infty} \cE_a(u_n)\ge \cE_a(u_0).
$$
Thus $u_0$ is a minimizer for $E(a)$. This ends the proof of Lemma \ref{lem:V3}. 
\end{proof}
The proof of Theorem \ref{thm:main1} is complete.

\section{Blow-up: weakly singular potentials} \label{sec:main2}

In this section we prove Theorem \ref{thm:main2}. We will always assume that $V$ satisfies \eqref{eq:weakV}, i.e.
$$
V\in L^1_{\rm loc}(\R^2), \quad V(x) \ge -C \sum_{j\in J} \frac{1}{|x-z_j|^{p}}, \quad 0<p<1.
$$

To simplify the notation, let us denote by $u_a$ the approximate minimizer for $E(a)$, i.e. $\cE_a(u_a)/E(a)\to 1$ and write $a\uparrow a^*$ instead of $a_n \uparrow a^*$. Also, we denote  
$$\ell_a=(a^*-a)^{-1}, \quad \beta = \frac{a^*}{2} \iint \frac{|Q_0(x)|^2 |Q_0(y)|^2}{|x-y|} \d x \d y.$$

We will always consider the case when $a$ is sufficiently close to $a^*$. We start with

\begin{lemma}[A-priori estimate] \label{lem:lem9} We have
$$
\quad C\ell_a^2 \ge \int |\nabla u_a|^2 \ge \frac{\ell_a^2}{C},\quad \int |\nabla u_a|^2 - \frac{a}{2} \int |u_a|^4 \le C\ell_a
$$
\end{lemma}

\begin{proof} From Lemma \ref{lem:1} and the assumption $\cE_a(u_a)/E(a)\to 1$ we have the sharp upper bound 
\bq \label{eq:weakV-Ea<=}
\cE_a(u_a)\le E(a)(1+o(1))\le -\frac{\beta^2}{a^*} \ell_a (1+o(1)). 
\eq

Now we go to the lower bound. We recall an elementary result, whose proof follows  by a simple scaling argument (see e.g. \cite[Proof of Lemma 6]{Phan-17a} for details).
\begin{lemma}\label{lem:Hyd} For every $0<q<2$, $y\in \R^2$ and $\eps>0$, we have
$$
\int \frac{|u(x)|^2}{|x-y|^q}\d x \le \eps\int |\nabla u|^2 + C_q \eps^{-q/(2-q)} \int |u|^2, \quad \forall u\in H^1(\R^2).
$$
\end{lemma}
Since $V$ satisfies \eqref{eq:weakV}, Lemma \ref{lem:Hyd} implies that
\bq \label{eq:Vu-lwb1}
\int V|u_a|^2 \ge - \eps\int |\nabla u_a|^2 - C_p \eps^{-p/(2-p)}, \quad \forall \eps>0 .
\eq
Combining with Lemma \ref{lem:cEa-lbs}, we find that
\bq \label{eq:lwbbbbbb}
\cE_a(u_a)\ge \Big(1-\frac{a}{a^*}-\eps \Big) \int |\nabla u_a|^2 -C\eps^{-1} - C_p \eps^{-p/(2-p)}, \quad \forall \eps>0.
\eq

From \eqref{eq:lwbbbbbb}, choosing 
$$\eps=\frac{1}{2}\Big( 1-\frac{a}{a^*}\Big)$$
and using $p<1$ (i.e. $\eps^{-p/(2-p)}\ll \eps^{-1}$ for $\eps>0$ small) we obtain the lower bound 
\begin{align*}
\cE(u_a) \ge \frac{1}{2}\Big( 1-\frac{a}{a^*}\Big)  \int |\nabla u_a|^2  - \frac{C}{a^*-a} \ge -C\ell_a.
\end{align*}
Comparing the latter estimate with the upper bound \eqref{eq:weakV-Ea<=} we also obtain 
\begin{align*}
\int_{\R^2}|\nabla u_a|^2 \le C \ell_a^{2}.
\end{align*}

On the other hand, in \eqref{eq:lwbbbbbb} we can choose 
$$\eps=\gamma \Big( 1-\frac{a}{a^*}\Big)$$
for a constant $\gamma>1$ to get
$$
\cE_a(u_a)\ge -(\gamma-1)\Big(1-\frac{a}{a^*} \Big) \int |\nabla u|^2 -\frac{C}{\gamma(a^*-a)}.
$$
If $\gamma$ is sufficiently large, we can use latter estimate and the upper bound \eqref{eq:weakV-Ea<=} to deduce  that 
$$
\int_{\R^2}|\nabla u_a|^2 \ge \frac{\ell_a^{2}}{C}.
$$

Finally, inserting the upper bound $\int |\nabla u|^2\le C\ell_a^2$ into \eqref{eq:Vu-lwb1} and \eqref{eq:VVVV}, then optimizing over $\eps>0$ we have
\bq \label{eq:Vuuuaaa}
-\int V|u_a|^2\le o(\ell_a), 
\eq
(we used $p<1$) and
\bq
\iint \frac{|u_a(x)|^2|u_a(y)|^2} {|x-y|} \d x \d y \le C\ell_a.
\eq
Combining with the upper bound $\cE_a(u_a)\le C\ell_a$ we conclude that
$$
\int |\nabla u_a|^2 -\frac{a}{2}\int |u_a|^4 = \cE_a(u_a) - \int V|u_a|^2 + \iint \frac{|u_a(x)|^2|u(y)|^2} {|x-y|} \d x \d y \le C\ell_a.
$$
\end{proof}

Now we are ready to give

\begin{proof}[Proof of Theorem \ref{thm:main2}] Denote 
\bq
u_a(x)=\ell_a  \varphi_a(\ell_a x).
\eq
Then $\|\varphi_a\|_{L^2}=1$ and by Lemma \ref{lem:lem9}, 
$$
\quad C \ge \int |\nabla \varphi_a|^2 \ge \frac{1}{C}, \quad \int |\nabla \varphi_a|^2 - \frac{a}{2} \int |\varphi_a|^4 \le C\ell_a^{-1}\to 0.
$$
Thus $\varphi_a$ is an optimizing sequence for the Gagliardo-Nirenberg inequality \eqref{eq:GN}. By Lemma \ref{lem:min-GN} in Appendix \ref{app:B}, there exist a subsequence of $\varphi_a$ (still denoted by $\varphi_a$ for simplicity), a sequence $\{x_a\}\subset \R^2$ and a constant $b>0$ such that 
$$
\varphi_a(x+x_a)\to bQ_0(bx)
$$
strongly in $H^1(\R^2)$. 

Now we determine $b$. Since $\varphi_a(x+x_a)\to bQ_0(bx)$ strongly in $H^1(\R^2)$, we have
$$
\int |\nabla \varphi_a|^2 \to b^2 \int |\nabla Q_0|^2 = b^2
$$
and
$$
\iint \frac{|\varphi_a(x)|^2|\varphi_a(y)|^2} {|x-y|} \d x \d y \to b \iint \frac{|Q_0(x)|^2|Q_0(y)|^2} {|x-y|} \d x \d y= \frac{2\beta b}{a^*}.
$$
Combining with \eqref{eq:Vuuuaaa} and the Gagliardo-Nirenberg inequality \eqref{eq:GN}  we have 
\begin{align*}
\cE_a(u_a)&\ge \Big( 1-\frac{a}{a^*}\Big) \int |\nabla u_a|^2 - \iint \frac{|u_a(x)|^2|u_a(y)|^2} {|x-y|} \d x \d y + o(\ell_a) \\
&= \ell_a^2 \Big( 1-\frac{a}{a^*}\Big)  \int |\nabla \varphi_a|^2 - \ell_a \iint \frac{|\varphi_a(x)|^2|\varphi_a(y)|^2} {|x-y|} \d x \d y+ o(\ell_a) \\
&= \frac{\ell_a}{a^*} \Big( b^2 - 2\beta b+o(1)\Big)= - \frac{\ell_a }{a^*} \Big( \beta^2 - (b-\beta)^2 + o(1)\Big).
\end{align*}
Comparing with the upper bound 
$$
\cE_a(u_a)\le -\frac{\ell_a}{a^*}  (\beta^2+o(1))
$$
in Lemma \ref{lem:lem9}, we conclude that $b=\beta$ and 
$$
\cE_a(u_a) = -\frac{\ell_a}{a^*}  (\beta^2+o(1)).
$$
Thus we have proved that
\bq \label{eq:cv-vpQ}
\varphi_a(x+x_a)\to \beta Q_0(\beta x)
\eq
strongly in $H^1(\R^2)$ which is equivalent to \eqref{eq:weakV-cVu} and
$$
E(a)= \cE_a(u_a) (1+o(1)) = -\frac{\ell_a}{a^*}  (\beta^2+o(1)).
$$
which is equivalent to \eqref{eq:Ea-weak}.

Finally, consider the cases when $V(x)=|x|^q$ for $q>0$ or $V(x)=-|x|^{-q}$ for $0<q<1$. Then by Theorem \ref{thm:main1}, $E(a)$ has a minimizer $u_a \ge 0$. Moreover, by the rearrangement inequalities \cite[Chapter 3]{LiLo-01}, we deduce that 
$$
\cE_a(u_a)\ge \cE_a(u_a^*)
$$
where $u_a^*$ is the symmetric decreasing rearrangement of $u_a$, and the equality occurs if and only if $u_a=u_a^*$ (since $|x|^q$ is strictly symmetric increasing and $|x|^{-q}$ is strictly symmetric decreasing). Since $u_a$ is a minimizer for $E(a)$, we conclude that $u_a$ must be radially symmetric decreasing. Consequently, $\varphi_a$ is also radially symmetric decreasing. From the convergence \eqref{eq:cv-vpQ} and the fact that $Q_0$ is radially symmetric decreasing, it is easy to deduce that $\varphi_a\to Q_0$ strongly in $H^1(\R^2)$. This is equivalent to \eqref{eq:limimim}. 

This ends the proof of Theorem \ref{thm:main2}.
\end{proof}

\section{Blow-up: strongly singular potentials} \label{sec:main3}

In this section we prove Theorem \ref{thm:main3}. Recall that $V$ satisfies assumption \eqref{eq:AS-V}, i.e.
$$
V(x)= - h(x) \sum_{j=1}^J  |x-z_j|^{- p_j}, \quad 0<p_j<2, \quad h\in C(\R^2), \quad C\ge  h \ge 0
$$
and 
$$p=\max_{j\in J} p_j\ge 1, \quad h_0=\max\{h(z_j):p_j=p\}>0.$$
Also we denote 
$$\mathcal{Z}=\{x_j: p_j=p, h(x_j)=h_0\}.$$ 

Again, we denote by $u_a$ the approximate minimizer for $E(a)$, i.e. $\cE_a(u_a)/E(a)\to 1$ and write $a\uparrow a^*$ instead of $a_n \uparrow a^*$. Also, we denote  
$$\ell_a=(a^*-a)^{-\frac{1}{2-p}}, \quad \beta = \left(\frac{pa^*A}{2}\right)^{\frac{1}{p-2}}$$
with $A$ defined in \eqref{eq:betaAAA}.

Since $V$ is sufficiently singular, the upper bound in Lemma \ref{lem:1} is not optimal when $a\uparrow a^*$. Instead, we have

\begin{lemma} \label{lem:Ea-up} We have 
\begin{align*}
\limsup_{a\uparrow a^*} E(a) \ell_a^{-p}  \le \inf_{\lambda>0} \left[ \frac{\lambda^2}{a^*} - \lambda^p A \right] = \frac{\beta^2}{a^*}-\beta^p A.
\end{align*}
\end{lemma}

\begin{proof} Let $x_j\in \mathcal{Z}$, i.e. $p_j=p$ and $h(x_j)=h_0$. For every $\eps>0$ there exists $\eta_\eps>0$ such that 
$$
V(x)\le (\eps - h_0) |x-x_j|^{-p}, \quad \forall  |x-x_j|\le 2\eta_\eps.
$$
We choose
$$
u_\ell(x)=A_{\ell} \varphi (x-x_j) \ell Q_0(\ell(x-x_j))
$$
where $0\le \varphi  \in C_c^\infty (\R^2)$ with $\varphi (x)=1$ for $|x|\le \eta_\eps$, $\varphi(x)=0$ for $|x|\ge 2\eta_\eps$, and $A_{\ell}>0$ is a suitable factor to make $\|u_\ell\|_{L^2}=1$. 

Similarly to the proof of Lemma \ref{lem:1}, since $Q_0$ and $|\nabla Q_0|$ decay exponentially we have 
\begin{align*}
A_\ell&=1+O(\ell^{-\infty}),\\
\int |\nabla u_\ell|^2 &= \ell^2 \int |\nabla Q_0|^2 + O(\ell^{-\infty}),\\
\int |u_\ell|^4&=\ell^2 \frac{2}{a^*}\int |\nabla Q_0|^2 + O(\ell^{-\infty})
\end{align*}
and
$$
\iint \frac{|u_\ell(x)|^2 |u_\ell(y)|^2}{|x-y|} \d x \d y = \ell \iint \frac{|Q_0(x)|^2 |Q_0(y)|^2}{|x-y|} \d x \d y + O(\ell^{-\infty}),
$$
Moreover, the choice of $\varphi$ ensures that
$$
V(x) \varphi (x-x_j)  \le  (\eps -  h_0) |x-x_j|^{-p} \chi_{\{|x-x_j|\le \eta_\eps\}}.
$$
Therefore,
\begin{align*}
\int_{\R^2} V |u_\ell|^2 &\le (\eps - h_0)  A_\ell^2 \int_{|x-x_j|\le \eta_\eps}  \frac{|Q_0(\ell(x-x_j))|^2}{|x-x_j|^{p}} \ell^2 \d x  \\
&= (\eps - h_0) \ell^{p} \int_{\R^2}  \frac{|Q_0(x)|^2}{|x|^{p}} \d x + O(\ell^{-\infty}).
\end{align*}
Putting all together, then using $E(a)\le \cE_a(u_\ell)$ and \eqref{eq:Q0} we obtain  
\begin{align*}
E(a) \le \ell^2 \left( 1- \frac{a}{a^*} \right) -  \ell \iint \frac{|Q_0(x)|^2 |Q_0(y)|^2}{|x-y|} \d x \d y +  (\eps - h_0) \ell^{p} \int_{\R^2}  \frac{|Q_0(x)|^2}{|x|^{p}} \d x + O(\ell^{-\infty}).
\end{align*}

Choosing $\ell=\lambda \ell_a=\lambda (a^*-a)^{-\frac{1}{2-p}}$ for a constant $\lambda>0$ we find that
\begin{align*}
E(a)\ell_a^{-p} \le \frac{\lambda^2}{a^*}- \lambda \ell_a^{1-p} \iint \frac{|Q_0(x)|^2 |Q_0(y)|^2}{|x-y|} \d x \d y+ \lambda^p (\eps - h_0) \int_{\R^2}  \frac{|Q_0(x)|^2}{|x|^{p}} \d x + O(\ell_a^{-\infty}).
\end{align*}
Taking the limit $\ell_a\to \infty$, using $p\ge 1$ and then taking $\eps\to 0$, we conclude that
\begin{align} \label{eq:Ea-up-up}
\limsup_{a\uparrow a^*} E(a)\ell_a^{-p} \le \frac{\lambda^2}{a^*} -\lambda^p A, \quad \forall \lambda>0.
\end{align}
It is straightforward to see that the right side of \eqref{eq:Ea-up-up} is smallest when 
$$
\lambda=\beta=\left(\frac{pa^*A}{2}\right)^{\frac{1}{p-2}}
$$
This ends the proof. 
 \end{proof}
 
 Next, we have

\begin{lemma}[A-priori estimates]\label{lem:EK} 
\begin{align*}
C \ell_a^2  \ge \int_{\R^2}|\nabla u_a|^2 \ge \frac{\ell_a^2}{C}, \quad \int |\nabla u_a|^2 -\frac{a}{2}\int |u_a|^4 \le C\ell_a^p, \quad \int V|u_a|^2 \le -\frac{\ell_a^p}{C}.
\end{align*}
\end{lemma}

\begin{proof} The proof strategy is similar to Lemma \ref{lem:lem9}. By the choice of $V$, we have 
$$
V(x) \ge-C \sum_{j\in J} \frac{1}{|x-z_j|^p} -C.
$$
Therefore, 
\bq \label{eq:Vu-lwb2}
\int V|u_a|^2 \ge - \eps\int |\nabla u_a|^2 - C_p \eps^{-p/(2-p)}, \quad \forall \eps>0 
\eq
by Lemma \ref{lem:Hyd} and 
\bq \label{eq:lwbb}
\cE_a(u_a)\ge \Big(1-\frac{a}{a^*}-\eps \Big) \int |\nabla u_a|^2 -C\eps^{-1} - C_p \eps^{-p/(2-p)}, \quad \forall \eps>0.
\eq
by Lemma \ref{lem:cEa-lbs}. Choosing 
$$\eps=\frac{1}{2}  \Big(1-\frac{a}{a^*} \Big)$$
and using the upper bound in Lemma \ref{lem:Ea-up} $p\in [1,2)$ we get (with  $p\in [1,2)$)
$$
-\frac{\ell_a^p}{C}  \ge E(a)(1+o(1)) \ge \cE_a(u_a) \ge C\ell_a^{p-2} \int |\nabla u_a|^2 - C_p \ell_a^p 
$$
Thus
$$
\int |\nabla u_a| \le C\ell_a^2. 
$$
Moreover, in \eqref{eq:lwbb} choosing  
$$
\eps=\gamma  \Big(1-\frac{a}{a^*} \Big)
$$
with $\gamma$ sufficiently large, we get
$$
\int |\nabla u_a| \ge \frac{\ell_a^2}{C}.
$$

Next, by \eqref{eq:cEu>-},
$$
\int_{\R^2} V |u_a|^2 \le \cE_a(u_a) - \left(1-\frac{a}{a^*}-\eps \right) \int_{\R^2} |\nabla u_a|^2 +C\eps^{-1}.
$$ 
Choosing again 
$$
\eps=\gamma  \Big(1-\frac{a}{a^*} \Big)
$$
with $\gamma$ sufficiently large, then using  $\cE_a(u_a)\le -\ell_a^p/C$ and $\int |\nabla u_a|^2\le C\ell_a^2$, we obtain
$$
\int_{\R^2} V |u_a|^2 \le -\frac{\ell_a^p}{C}.
$$

Finally, inserting the bound $
\int |\nabla u_a| \le C\ell_a^2
$ into \eqref{eq:Vu-lwb2} and \eqref{eq:VVVV}, then optimizing over $\eps>0$ we have
$$
\int V|u_a|^2 \ge -C\ell_a^p, \quad \iint \frac{|u_a(x)|^2 |u_a(y)|^2}{|x-y|} \d x \d y \le C \ell_a
$$
Thus
$$
\int |\nabla u_a|^2 -\frac{a}{2}\int |u_a|^4 = \cE_a(u_a)-\int V|u_a|^2 + \iint \frac{|u_a(x)|^2 |u_a(y)|^2}{|x-y|} \d x \d y \le C \ell_a^p.
$$
\end{proof}

Now we go to

\begin{proof}[Proof of Theorem \ref{thm:main3}] From the estimate $\int V|u_a|^2\le -\ell_a^p/C$ in Lemma \ref{lem:EK} and the simple bound 
$$
V(x) \ge -  C \sum_{j\in J} \frac{1}{|x-z_{j}|^{p_j}} 
$$
we obtain that, up to a subsequence of $u_a$, there exists $i\in \{1,2,...,J\}$ such that 
\bq \label{eq:varphi-V}
\int_{\R^2} \frac{|u_a(x)|^2}{ |x-z_i|^{p_{i}}} \d x \ge \frac{\ell_a^p}{C}.
\eq
Define  
\bq
u_a(x)=\ell_a \varphi_a(\ell_a (x -z_{i})).
\eq
By Lemma \ref{lem:EK}, $\varphi_a$ is bounded in $H^1(\R^2)$, $\|\varphi_a\|_{L^2}=1$ and
$$
\int |\nabla \varphi_a|^2 \ge \frac{1}{C}, \quad \int |\nabla \varphi_a|^2 - \frac{a}{2} \int |\varphi_a|^4 \le C\ell_a^{p-2} \to 0.
$$
Thus $\varphi_a$ is an optimizing sequence for the Gagliardo-Nirenberg inequality \eqref{eq:GN}. By Lemma \ref{lem:min-GN} in Appendix \ref{app:B}, up to a subsequence of $\varphi_a$, there exist a sequence $\{x_a\}\subset \R^2$ and a constant $b>0$ such that 
$$
\varphi_a(x+x_a)\to bQ_0(bx)\quad \text{strongly in~}H^1(\R^2).
$$

Next, we prove $ p_i=p$. From \eqref{eq:varphi-V}, we have 
$$
\ell_{a}^{p_i}\int_{\R^2} \frac{|\varphi_a(x)|^2}{ |x|^{p_i}} \d x \ge \frac{\ell_a^p}{C}.
$$
This implies that $ p_i=p$ because $\ell_a\to \infty$ and $\int |\varphi_a|^2/|x|^{p_i}$ is bounded (as $\varphi_a$ is bounded in $H^1(\R^2)$). 

\medskip

Moreover, since $\varphi_a(x+x_a)\to bQ_0(bx)$ in $H^1(\R^2)$ and $1\le p<2$, we find that
$$
\frac{1}{C} \le \int \frac{|\varphi_a(x)|^2}{ |x|^{p}} \d x =\int \frac{|\varphi_a(x+x_a)|^2}{ |x+x_a|^{p}} \d x  = \int \frac{|bQ_0(bx)|^2}{ |x+x_a|^{p}} \d x + o(1).
$$
Since $Q_0$ decays exponentially, we conclude that $x_a$ is bounded. Up to a subsequence, we can assume that $x_a\to x_\infty$. Thus we have
$$
\varphi_a(x+x_\infty)\to bQ_0(bx)\quad \text{strongly in~} H^1(\R^2).
$$

Finally, we determine $x_\infty$ and $b$. This will be done by considering the sharp lower bound for $E(a)$. By the assumptions on $V$ and $p=p_i\ge p_j$, we have
$$
V(x) \ge -\frac{\eps+h(x_i)}{|x-z_i|^{p}} - C \sum_{j\ne i} \frac{1}{|x-z_j|^{p}} - C_\eps, \quad \forall \eps>0.
$$
Therefore, using $u_a(x)=\ell_a \varphi_a(\ell_a (x -z_{i}))$ we get
\begin{align*}
\int V|u_a|^2 &= \int V(\ell_a^{-1}x+z_i) |\varphi_a(x)|^2 \d x\\
&\ge - \ell_a^p (\eps+h(x_i)) \int \frac{|\varphi_a(x)|^2}{|x|^p} \d x + \sum_{j\ne i} \int \frac{|\varphi_a(x)|^2}{|\ell_a^{-1} x + z_i-z_j|^p} \d x -C_\eps.
\end{align*}
Using $\varphi_a(x)\to bQ_0(b(x-x_\infty))$ in $H^1(\R^2)$ and $Q_0$ decays exponentially, we  have
$$\int \frac{|\varphi_a(x)|^2}{|x|^p} \d x= b^p \int \frac{|Q_0(x-bx_\infty)|^2}{|x|^p} \d x+ o(1)$$
and 
$$
\int \frac{|\varphi_a(x)|^2}{|\ell_a^{-1} x + z|^p} \d x = o(\ell_a^p), \quad \forall z\ne 0.
$$
Thus
\begin{align} \label{fn1}
\int V|u_a|^2 \ge - \ell_a^pb^p (\eps+h(x_i)) \int \frac{|Q_0(x-bx_\infty))|^2}{|x|^p} \d x + o(\ell_a^p)-C_\eps, \quad \forall \eps>0.
\end{align}

Moreover, using again $u_a(x)=\ell_a \varphi_a(\ell_a (x -z_{i}))$ and $\varphi_a(x)\to bQ_0(b(x-x_\infty))$ in $H^1(\R^2)$, we have
\bq  \label{fn2}
\int |\nabla u_a|^2= \ell_a^2 b^2 \int |\nabla Q_0|^2 + o(\ell_a^2) =\ell_a^2 b^2 +o(\ell_a^2)+o(1)
\eq
and
\bq  \label{fn3}
\iint \frac{|u_a(x)|^2 |u_a(y)|^2}{|x-y|} \d x \d y = \ell_a b \iint \frac{|Q_0(x)|^2 |Q_0(y)|^2}{|x-y|} \d x \d y + o(\ell_a).
\eq

In summary, from the Gagliardo-Nirenberg inequality \eqref{eq:GN} and \eqref{fn1}, \eqref{fn2}, \eqref{fn3}  we have
\begin{align*}
\cE_a(u_a)&\ge \Big( 1-\frac{a}{a^*}\Big) \int |\nabla u_a|^2 + \int V|u_a|^2 - \iint \frac{|u_a(x)|^2|u_a(y)|^2} {|x-y|} \d x \d y \\
&= \frac{ \ell_a^p b^2}{a^*} - \ell_a^p b^p (\eps+h(x_i)) \int \frac{|Q_0(x-bx_\infty)|^2}{|x|^p} \d x \\
&\quad - \ell_a b \iint \frac{|Q_0(x)|^2 |Q_0(y)|^2}{|x-y|} \d x \d y +o(\ell_a^p)-C_\eps.
\end{align*}
Here we have used $p\ge 1$, so that $\ell_a^p\ge \ell_a$. Taking the limit $\ell_a\to \infty$, then taking $\eps\to 0$ and using the assumption $\cE_a(u_a)/E(a)\to 1$ we obtain
\begin{align} \label{eq:Ea-lb-sharp}
\liminf_{a\to a^*}E(a)\ell_a^{-p}& \ge \frac{b^2}{a^*} - b^p B 
\end{align}
where
\begin{equation} \label{eq:betaB}
B=
\begin{cases} 
\, h(x_i)\int_{\R^2}  \frac{|Q_0(x-bx_\infty)|^2}{|x|^{p}} \d x + \iint \frac{|Q_0(x)|^2 |Q_0(y)|^2}{|x-y|} \d x \d y , &\quad \text{if}\,\, p=1\\
\,h(x_i) \int_{\R^2}  \frac{|Q_0(x-bx_\infty)|^2}{|x|^{p}} \d x. &\quad \text{if}\,\, p>1 \\
\end{cases}
\end{equation}

Finally, using $h(x_i)\le h_0$ and the rearrangement inequality, we have  
$$
h(x_i) \int_{\R^2}  \frac{|Q_0(x-bx_\infty)|^2}{|x|^{p}} \d x \le h_0 \int_{\R^2}  \frac{|Q_0(x)|^2}{|x|^{p}} \d x 
$$
where the equality occurs if and only if $h(x_i)=h_0$ and $x_\infty=0$ (here $Q_0$ is symmetric decreasing and $|x|^{-p}$ is strictly symmetric decreasing). Thus $B\le A$ and hence \eqref{eq:Ea-lb-sharp} implies that
$$
\liminf_{a\to a^*}E(a)\ell_a^{-p} \ge \frac{b^2}{a^*} - b^p A. 
$$
Comparing the latter estimate with the upper bound in Lemma \ref{lem:Ea-up}, we conclude that $b=\beta$ and $A=B$, i.e. $h(x_i)=h_0$ and $x_\infty=0$. Thus $x_i\in \mathcal{Z}$ and 
$$
\varphi_a(x)\to \beta Q_0(\beta x) \quad \text{strongly in~} H^1(\R^2)
$$ 
which is equivalent to \eqref{eq:cv-p>=1}. If $\mathcal{Z}$ has a unique element, then we obtain the convergence for the whole sequence $u_a$ by a standard argument. 

The proof is complete.
\end{proof}

\appendix 

\section{Appendix} \label{sec:app}

\subsection{Concentration-compactness lemma}

For the reader's convenience, we recall the following fundamental result of Lions \cite{Lions-84,Lions-84b}.

\begin{lemma}[Concentration-compactness]\label{lem:concom} Let $N\ge 1$. Let $\{u_n\}$ be a bounded sequence in $H^1(\R^N)$ with $\|u_n\|_{L^2}=1$. Then there exists a subsequence (still denoted by $\{u_n\}$ for simplicity) such that  one of the following cases occurs:

\begin{itemize}

\item[(i)] {\rm (Compactness)} There exists a sequence $\{z_n\}\subset \R^N$ such that $u_n(.+z_n)$ converges to a function $u_0$ weakly in $H^1(\R^N)$ and strongly in $L^p(\R^N)$ for all $p\in [2,2^*)$. 
\item[(ii)] {\rm (Vanishing)} $u_n\to 0$ strongly in $L^p$ for all $p\in (2,2^*)$. 
\item[(iii)] {\rm (Dichotomy)} There exist $\lambda\in (0,1)$ and two sequences $\{u_n^{(1)}\}$, $\{u_n^{(2)}\}$ in $H^1(\R^N)$ such that 
\[ 
\left \{
\begin{aligned}
  & \lim_{n\to \infty} \int_{\R^N} |u_n^{(1)}|^2 = \lambda,  \quad \lim_{n\to \infty} \int_{\R^N} |u_n^{(2)}|^2 = 1-\lambda,  \\
  & \lim_{n\to \infty} {\rm dist} ({\rm supp} (u_n^{(1)}), {\rm supp} (u_n^{(2)})) =+\infty;\\
  & \lim_{n\to \infty} \|u_n - u_n^{(1)}-u_n^{(2)}\|_{L^p} = 0  , \quad  \forall p\in [2,2^*); \\
  & \liminf_{n\to \infty} \int_{\R^N} (|\nabla u_n|^2 - |\nabla u_n^{(1)}|^2 - |\nabla u_n^{(2)}|^2) \ge 0.
  \end{aligned}
 \right.
\]
\end{itemize}

\end{lemma}
Here $2^*$ is the critical power in Sobolev's embedding, i.e. $2^*=2N/(N-2)$ if $N\ge 3$ and $2^*=+\infty$ if $N\le 2$.  
\begin{proof} The result is essentially taken from  \cite[Lemma III.1]{Lions-84}, with some minor modifications that we will explain.

First, the original notion of the compactness case in \cite[Lemma I.1]{Lions-84} reads
\bq \label{eq:compactness}
\lim_{R\to \infty} \int_{|x|\le R} |u_n(x+x_n)|^2 \d x =1.
\eq
However, since $u_n$ is bounded in $H^1(\R^N)$, the statement (i) follows from \eqref{eq:compactness} easily. 

Next, the original notion of the vanishing case in \cite[Lemma I.1]{Lions-84} reads
$$
\lim_{R\to \infty} \sup_{y\in \R^N} \int_{|x|\le R} |u_n(x+y)|^2 \d x =0.
$$
This and the boundedness in $H^1$ implies that $u_n\to 0$ strongly in $L^p(\R^N)$ for all $p\in (2,2^*)$, as explained in \cite[Lemma I.1]{Lions-84b}.

Finally, in the dichotomy case, (iii) follows from \cite[Lemma III.1]{Lions-84} (the original statement has a parameter $\eps\to 0$) and a standard Cantor's diagonal argument. 
\end{proof} 

\subsection{Optimizing sequences of Gagliardo-Nirenberg inequality} \label{app:B}

In this section we recall 
\begin{lemma}[Compactness of optimizing sequences for Gagliardo-Nirenberg inequality] \label{lem:min-GN} If $f_n \ge 0$ is a bounded sequence in $H^1(\R^2)$, $\|f_n\|_{L^2}=1$ and 
$$
\int |\nabla f_n|^2 -\frac{a^*}{2}\int |f_n|^4 \to 0, \quad \liminf \|\nabla f_n\|_{L^2} >0,
$$
then up to subsequences and translations, $f_n$ converges strongly in $H^1(\R^2)$ to $bQ_0(bx)$ for some constant $b>0$.
\end{lemma}

This is a standard result, see e.g. \cite{BFV-14} and the references theirin. Since the result can be proved easily by the Concentration-Compactness Lemma \ref{lem:concom}, let us provided it below for the reader's convenience. 
\begin{proof} We apply  to the sequence $f_n$.

\medskip

{\bf No-vanishing.} If the vanishing case occurs, then $\|f_n\|_{L^4} \to 0$, but this contradicts to the assumption  in Lemma \ref{lem:min-GN}.

\medskip

{\bf No-dichotomy.} Assume that the dichotomy case occurs. Then we can find two sequences $f_n^{(1)}$, $f_n^{(2)}$ such that
$$
\int |f_n^{(1)}|^2 \to \lambda \in (0,1), \quad \int |f_n^{(2)}|^2\to 1-\lambda,
$$
$$
\int |\nabla f_n|^2 \ge  \int |\nabla f_n^{(1)}|^2  + \int |\nabla f_n^{(2)}|^2 + o(1)
$$
and
$$
\int |f_n|^4 =  \int |f_n^{(1)}|^4 + \int |\nabla f_n^{(2)}|^4+o(1).
$$

On the other hand, using the Gagliardo-Nirenberg inequality \eqref{eq:GN} for $f_n^{(j)}/\|f_n^{(j)}\|_{L^2}$  we obtain
$$
\int |\nabla f_n^{(1)}|^2 \ge \frac{a^*}{2\lambda} \int |f_n^{(1)}|^4, \quad \int |\nabla f_n^{(2)}|^2 \ge \frac{a^*}{2(1-\lambda)} \int |f_n^{(2)}|^4 .
$$ 
Thus
\begin{align*}
\int |\nabla f_n|^2 &\ge   \int |\nabla f_n^{(1)}|^2  + \int |\nabla f_n^{(2)}|^2 + o(1) \\
&\ge \frac{a^*}{2\lambda} \int |f_n^{(1)}|^4 +  \frac{a^*}{2(1-\lambda)} \int |f_n^{(2)}|^4 + o(1)\\
&\ge \frac{a^*}{2} \max \Big( \frac{1}{\lambda}, \frac{1}{1-\lambda} \Big) \Big(  \int |f_n^{(1)}|^4 + \int |\nabla f_n^{(2)}|^4  \Big) +o(1)\\
&=  \frac{a^*}{2} \max \Big( \frac{1}{\lambda}, \frac{1}{1-\lambda} \Big) \int |f_n|^4 + o(1).
 \end{align*}
However, this again contradicts to the assumption in Lemma \ref{lem:min-GN} since $\lambda\in (0,1)$. Thus the dichotomy case does not occur. 

\medskip
{\bf Compactness.} Now up to subsequences and translations, $f_n$ converges to a function $f$ weakly in $H^1(\R^2)$ and strongly in $L^s(\R^2)$ for all $s\in [2,\infty)$. Therefore, $\|f\|_{L^2}=1$ and
$$ \int |\nabla f_n|^2 \ge \int |\nabla f|^2+ o(1), \quad  \int |f_n|^4 \to \int |f|^4 . $$
On the other hand, by the assumption in Lemma \ref{lem:min-GN}, we have
$$
 \int |\nabla f|^2 - \frac{a^*}{2} \int |f|^4 \le \int |\nabla f_n|^2 - \frac{a^*}{2}  \int |f_n|^4  \to 0.
$$
Thus $f$ is an optimizer for the Gagliardo-Nirenberg inequality \eqref{eq:GN}. Moreover, we must have $\int |\nabla f_n|^2 \to \int |\nabla f|^2$, i.e. $f_n$ converges to $f$ strongly in $H^1(\R^2)$.

Finally since $Q_0$ is the unique optimizer for \eqref{eq:GN} up to translations and dilations, we have $f(x+z_0)=bQ_0(bx)$ for constants $b>0$ and $z_0\in \R^2$. Thus $f_n(x+z_0)\to f(x+z_0)=bQ_0(bx)$ in $H^1(\R^2)$. This ends the proof.
\end{proof}

\end{document}